\documentclass[conference,10pt]{IEEEtran}

\usepackage[margin=1in]{geometry}

\usepackage{amsmath,amssymb}
\usepackage{amsthm}

\usepackage{graphicx}
\usepackage{booktabs}
\usepackage{caption}
\usepackage{subcaption}
\usepackage{xcolor}

\usepackage{cite}

\usepackage[hidelinks]{hyperref}

\theoremstyle{plain}
\newtheorem{theorem}{Theorem}
\newtheorem{lemma}{Lemma}
\newtheorem*{lemma*}{Lemma}

\theoremstyle{definition}
\newtheorem{definition}{Definition}

\theoremstyle{remark}
\newtheorem{remark}{Remark}

\title{Local Differential Privacy with Correlated Noise Achieves Central-DP Optimal Cost}

\author{
\IEEEauthorblockN{Madhura Pathegama, Srikanth Avasarala,
Viveck R. Cadambe,
Juba Ziani}
\IEEEauthorblockA{Georgia Institute of Technology, Atlanta, GA, USA\\
Email: \texttt{\{macharige3, savasarala, viveck, jziani3\}@gatech.edu}}
}

\begin{document}
\maketitle

\begin{abstract}
We study privately estimating the sum of $n$ user-held values in the presence of an honest-but-curious server. This motivates requiring privacy not only at data release but also throughout server-side computation. We therefore adopt the local (pure) differential privacy model, in which each user transmits a noise-perturbed value. It is well known that independent local noise typically incurs a substantial utility loss compared to the centralized model, where noise is added only after aggregation.

We show that this gap is not fundamental. By carefully designing correlations among the locally added noise variables, we construct $\varepsilon$-DP mechanisms whose estimation cost matches the optimal cost achievable in the centralized setting, up to an arbitrarily small error.
\end{abstract}

\section{Introduction}

Consider $n$ users indexed by $i\in\{1,\dots,n\}$, where user $i$ holds a private datum $x_i\in\mathbb{R}$.\footnote{Our analysis extends naturally to vector-valued $x_i$ under mild conditions; see the conclusion.} 
A server aims to compute and release the aggregate sum $\sum_{i=1}^n x_i$ (or the average $\frac1n\sum_{i=1}^n x_i$) while preserving the privacy of each individual datum.
This aggregation task arises in applications such as distributed and federated learning \cite{abadi2016deep,mcmahan2017communication}, consensus algorithms \cite{huang2012differentially,nozari2015differentially}, and sensor network computations \cite{rastogi2010differentially,won2014proactive}, where users contribute local updates or statistics that must be combined while preserving individual privacy.

In the \emph{central} differential privacy (central-DP) model \cite{dwork2006calibrating}, users transmit their raw data to a trusted server, which then adds noise to the aggregate before releasing it. When the server is not trusted (e.g., it may be honest-but-curious), stronger safeguards are required.

As an alternative, \emph{local} differential privacy (local-DP) was proposed \cite{kasiviswanathan2011can,duchi2013local}. In this model, privacy is enforced before data leaves the user. We consider an additive local-noise mechanism in which each user transmits $x_i + Z_i,$
where $Z_i$ is a random noise variable. The server’s estimate of the sum becomes
\[
\sum_{i=1}^n x_i + \sum_{i=1}^n Z_i.
\]
Because noise is added locally, one expects that the server cannot reliably infer any individual $x_i$. 

In this work, we quantify privacy using pure $\varepsilon$-differential privacy ($\varepsilon$-DP) and focus exclusively on additive noise mechanisms.
Let $x,\widetilde x\in\mathbb{R}^n$ be \emph{neighboring} data vectors, meaning they differ in at most one coordinate and the change is bounded by a sensitivity parameter $\Delta>0$; i.e., there exists $i\in\{1,\dots,n\}$ such that $x_j=\widetilde x_j$ for all $j\neq i$ and $|x_i-\widetilde x_i|\le\Delta$. The $\varepsilon$-DP is then defined as follows.

\begin{definition}\label{def: epsDP}
The additive mechanism $x\mapsto x+Z$ satisfies $\varepsilon$-DP for sensitivity $\Delta$ if, for all neighboring $x,\widetilde x\in\mathbb{R}^n$ and all measurable sets $A\subseteq\mathbb{R}^n$,
\begin{align}\label{eq: epsDP0}
\Pr\!\big(x+Z\in A\big)
\;\le\;
e^{\varepsilon}\Pr\!\big(\widetilde x+Z\in A\big).
\end{align}
\end{definition}

We seek to design the joint noise vector $Z^n=(Z_1,\dots,Z_n)$ to satisfy \eqref{eq: epsDP0} while minimizing the aggregate cost, which in our setting depends only on the noise:
\begin{align}\label{eq: cost}
    \mathbb{E}\!\left[\, L\!\left(\sum_{i=1}^n Z_i\right)\right].
\end{align}
Here $L(\cdot)$ is even and nondecreasing on $\mathbb{R}_+$. A common choice is the $L_p$ loss $L(x)=|x|^p$ for $p\ge 1$.

The prevailing approach in the local differential privacy literature is to add independent noise at each user, often i.i.d.\ from a common distribution \cite{kasiviswanathan2011can,duchi2013local}. Independence greatly simplifies the privacy constraints: when the $Z_i$ are independent, satisfying \eqref{eq: epsDP0} jointly is equivalent to each $Z_i$ satisfying the corresponding one-dimensional version of \eqref{eq: epsDP0}. 


However, independent local noise typically incurs a strict utility gap relative to the centralized model \cite{duchi2013local}. 
In particular, for quadratic loss $L(x)=x^2$, the optimal independent mechanism is obtained by choosing each $Z_i$ according to the one-dimensional optimal distribution, and the resulting aggregate cost is $n$ times the cost of the corresponding centrally added-noise mechanism.

Our goal is to mitigate this utility gap by allowing correlated noise across users. 
In our setting, such correlations may be established in an offline phase or via secure channels between users before the privatized data are sent to the server. Since the noise is still added on the user side, we refer to this as a local mechanism.
Note that Definition~\ref{def: epsDP} does not require independence and continues to guarantee privacy under correlation. In particular, letting $Z_{\sim i}:=(Z_1,\dots,Z_{i-1},Z_{i+1},\dots,Z_n)$, for any measurable $A_i\subseteq\mathbb{R}$ and $B\subseteq\mathbb{R}^{n-1}$, conditioning \eqref{eq: epsDP0} on the event $\{Z_{\sim i}\in B\}$ yields
\begin{align}\label{eq: condDP}
\Pr\big(x_i+&Z_i\in A_i \mid Z_{\sim i}\in B\big)\nonumber\\
&\;\le\;
e^\varepsilon
\Pr\big(\widetilde x_i+Z_i\in A_i \mid Z_{\sim i}\in B\big).
\end{align}
Thus, even if the server knows all other users' data and noise values, it still cannot reliably distinguish whether user $i$'s datum was $x_i$ or $\widetilde x_i$.

This shows that using correlated noise additions does not compromise the privacy during the computation phase as long as the mechanism satisfies $\varepsilon$-DP condition in \eqref{eq: epsDP0}. Against this backdrop, we ask the following question.

\emph{Can we design $\varepsilon$-DP additive local mechanisms with correlated noise that significantly outperform the standard independent-noise mechanisms?}
We show that the answer is affirmative. We prove that appropriately designed correlated local mechanisms can achieve costs arbitrarily close to the optimal cost attainable under centralized $\varepsilon$-DP. 

Our approach leverages the characterization of the optimal one-dimensional $\varepsilon$-DP mechanism in \cite{geng2015optimal}: we reduce the $n$-user problem to a one-dimensional optimization along the aggregate noise direction, obtain a one-dimensional lower bound, and show it is essentially tight via a correlated construction.

For basic sum estimation under local differential privacy, prior work \cite{joseph2019role} shows that protocols with fully interactive communication, where all messages between users and the server are visible to the server, cannot essentially improve the accuracy of standard noninteractive local mechanisms. Our work does not contradict this, as we allow users to establish correlations securely (e.g., offline or over secure channels\footnote{Relying on standard computational assumptions, such channels can be realized over communication through the server \cite{diffie1976new}.}) without revealing the underlying randomness to the server, which enables strictly stronger utility guarantees.


Note that our results have direct implications for distributed summation tasks, most notably distributed mean estimation \cite{duchi2013local,asi2022optimal}, which underlies aggregation in gradient-based optimization and learning \cite{rajkumar2012differentially,song2013stochastic} as well as consensus algorithms \cite{huang2012differentially}. In particular, under quadratic loss, correlated noise can achieve an $O(1)$ error for sum estimation, whereas independent local noise leads to an $O(n)$ scaling of the loss.

\subsection{Related work}

It is worth noting that under \emph{approximate} differential privacy%
\footnote{Approximate differential privacy, also known as $(\varepsilon,\delta)$-DP, relaxes \eqref{eq: epsDP0} to $\Pr(x+Z\in A) \le e^\varepsilon \Pr(\widetilde x+Z\in A) + \delta$.}, gains from correlated noise are known in the \emph{Gaussian} setting \cite{imtiaz2018distributed,imtiaz2021correlated,allouah2024privacy,vithana2025correlated,vithana2025differentially}. 
In particular, \cite{allouah2024privacy} and \cite{vithana2025differentially} show that correlated Gaussian noise can match the performance of centrally added Gaussian noise under quadratic loss. However, these results rely on structural properties specific to Gaussian mechanisms and are tailored to quadratic loss. They do not extend to pure $\varepsilon$-DP or to general loss functions.

Our focus in this work is to show that carefully constructed correlated local noise can achieve centralized-DP utility. We do not address the efficient or secure generation of such correlations; these may be developed based on secure aggregation techniques \cite{bonawitz2017practical,chen2022fundamental,stevens2022efficient,vithana2025correlated}.

It is also worth noting that temporal correlations in added noise have been used to improve distributed gradient methods \cite{kairouz2021practical,pillutla2025correlated}, but such correlations are not relevant to our setting.

\section{Preliminaries}

\subsection{Notation}

We use the following notation throughout: for $A\subseteq\mathbb{R}^n$ and $v\in\mathbb{R}^n$, $A+v:=\{x+v:\,x\in A\}$, and for $\alpha\in\mathbb{R}$, $\alpha A:=\{\alpha x:\,x\in A\}$.

For a function $f:\mathbb{R}^n\to\mathbb{R}$, we write $f(\alpha(\cdot))$ to denote the function $x\mapsto f(\alpha x)$.
Similarly for a probability measure $P$ on $\mathbb{R}^n$, we write $P(\alpha(\cdot))$ for $P(\alpha A)$ as a function of $A$. All probability measures considered in this work are Borel measurable.

\subsection{Differential privacy}

In this work, we consider additive mechanisms in which the noise variables $Z_i$ are chosen independently of the corresponding data $x_i$. Under this assumption, the privacy requirement \eqref{eq: epsDP0} depends only on the joint distribution $P$ of the noise vector $Z$, and the mechanism is fully specified by $P$.

To make this observation explicit, define the set of allowable coordinate-wise perturbations \vspace{-1mm}
\[
T_n^\Delta
\;:=\;
\bigcup_{i=1}^n \{0\}^{i-1}\times[-\Delta,\Delta]\times \{0\}^{n-i},
\]
i.e., the zero-centered union of length-$2\Delta$ line segments along each coordinate axis. Then $x,\widetilde x\in\mathbb{R}^n$ are neighboring (for sensitivity~$\Delta$) if and only if $\widetilde x=x+t$ for a unique $t\in T_n^\Delta$. Hence \eqref{eq: epsDP0} is equivalent to requiring that, for all $x\in\mathbb{R}^n$, all $t\in T_n^\Delta$, and all measurable $A\subseteq\mathbb{R}^n$,
\[
\Pr\!\big(x+Z\in A\big)
\;\le\;
e^{\varepsilon}\Pr\!\big(x+t+Z\in A\big).
\]
Equivalently \vspace{-2 mm}
\begin{align}\label{eq: epsDP}
P(A)\;\le\; e^{\varepsilon}\,P(A+t).
\end{align}

Accordingly, when the mechanism $x\mapsto x+Z$ is $\varepsilon$-DP for sensitivity $\Delta$, we will also say (by abuse of notation) that the noise law $P$ itself is $\varepsilon$-DP for sensitivity $\Delta$.

\begin{definition}\label{def: epsDP_noise}
Given $\varepsilon>0$ and $\Delta>0$, a probability measure $P$ on $\mathbb{R}^n$ satisfies $\varepsilon$-DP for sensitivity $\Delta$ if, for all measurable $A\subseteq\mathbb{R}^n$ and all $t\in T_n^\Delta$, \eqref{eq: epsDP} is satisfied.

\end{definition}

\subsection{Cost function}

The utility in our setting is quantified by the cost in \eqref{eq: cost}. If the noise vector $Z=(Z_1,\dots,Z_n)$ is distributed according to $P$  then the cost can be written as
\[
\int_{\mathbb{R}^n} L\!\left(\sum_{i=1}^n z_i\right)\,P(dz_1\cdots dz_n).
\]

We impose the following mild and natural assumptions on the loss function $L:\mathbb{R}\to\mathbb{R}$: \textbf{(A1)} $L$ is even; \textbf{(A2)} $L$ is nondecreasing on $\mathbb{R}_+$; and \textbf{(A3)} $L$ has at most subexponential growth, i.e.,
\vspace{-1mm}
\[
L(x)=\exp(o(|x|))\qquad \text{as } |x|\to\infty.
\]

Since both privacy and utility are determined by the noise law $P$ of $Z=(Z_1,\dots,Z_n)$, our focus is to solve the following optimization problem. 
\begin{align}\label{eq: opt_problem}
\inf_{P}
\int_{\mathbb{R}^n} L\!\left(\sum_{i=1}^n z_i\right)\,P(dz),
\end{align}
where optimization is over all $P$ satisfying $\varepsilon$-DP for sensitivity $\Delta$.
Our main result (Theorem~\ref{thm: main}) resolves \eqref{eq: opt_problem} by matching its optimum to the centralized benchmark and showing it is approachable under local DP using correlated noise. 

Although the cost in \eqref{eq: opt_problem} is defined in terms of an $n$-dimensional noise measure $P$, it is useful to also consider costs associated with one-dimensional probability measures, as much of the analysis reduces to a scalar problem. Accordingly, for a cost function $L:\mathbb{R}\to\mathbb{R}$ and a probability measure $Q$ on $\mathbb{R}$, we define the associated cost
\[
C(L,Q):= \int_{\mathbb{R}} L(x)\,Q(dx).
\]

Let $\mathcal{P}_n(\varepsilon,\Delta)$ denote the class of Borel probability measures on $\mathbb{R}^n$ that satisfy $\varepsilon$-differential privacy for sensitivity $\Delta$. We define the optimal cost (over one-dimensional noise measures) as
\begin{align}\label{eq: minimize}
    C^*_{\varepsilon,\Delta}(L)
    \;\triangleq\;
    \inf_{Q \in \mathcal{P}_1(\varepsilon,\Delta)} C(L,Q).
\end{align}

It was shown in \cite{geng2015optimal} that for any cost function $L$ satisfying Assumptions \textbf{(A1)}--\textbf{(A3)}, an optimal noise distribution exists and admits a density, which we denote by $f^*_{L,\varepsilon,\Delta}$. The optimal value of the objective can therefore be written as \vspace{-1mm}
\begin{align*}
    C^*_{\varepsilon,\Delta}(L)
    \;=\;
    \int_{\mathbb{R}} L(x)\, f^*_{L,\varepsilon,\Delta}(x)\, dx .
\end{align*}

More specifically, \cite{geng2015optimal} shows that the optimal density $f^*_{L,\varepsilon,\Delta}$ belongs to the parametric family of \emph{staircase mechanisms} $\{f_\gamma^{(\varepsilon,\Delta)}:\gamma\in[0,1]\}$. We use these functions to prove the achievability part of our main theorem; see Appendix~\ref{app: ach} for details.

\begin{remark}\label{rm: central}
The quantity $C^*_{\varepsilon,\Delta}(L)$ has a direct operational meaning in the centralized model of differential privacy. In the \emph{central} DP setting, users send their raw data to a trusted server, which releases a randomized aggregate. For the additive mechanism $\sum_{i=1}^n x_i+Z_0$, where $Z_0\sim Q$, the $\varepsilon$-DP requirement for neighboring $x,\widetilde x$ is
\[
\Pr\!\left(\sum_{i=1}^n x_i + Z_0 \in A\right)
\le
e^\varepsilon
\Pr\!\left(\sum_{i=1}^n \widetilde x_i + Z_0 \in A\right).
\]
Equivalently, the one-dimensional noise law $Q$ must satisfy $\varepsilon$-DP with sensitivity $\Delta$ in the sense of Definition~\ref{def: epsDP_noise}. Hence $C^*_{\varepsilon,\Delta}(L)$ is exactly the minimum achievable centralized cost. 
\end{remark}

We state a simple scaling property of the optimal cost under dilation of the loss function, which will be used in the proof of the main theorem. The proof is deferred to Appendix~\ref{app: lemma}.

\begin{lemma}[Scaling invariance of the optimal cost]\label{lem: scaling_cost}
Let $\varepsilon,\Delta,\alpha>0$ and let $L:\mathbb{R}\to\mathbb{R}$ satisfy \textbf{(A1)}--\textbf{(A3)}. Then
\begin{align}\label{eq: scaling_cost}
 C^*_{\varepsilon,\Delta/\alpha}(L(\alpha(\cdot))) = C^*_{\varepsilon,\Delta}(L).
\end{align}
\end{lemma}

\section{The main result}

Our main result characterizes the optimal cost achievable under $\varepsilon$-DP for sensitivity $\Delta$. The theorem is stated in two parts: a universal lower bound and an achievability statement.

\begin{theorem}\label{thm: main}
Let $\varepsilon,\Delta>0$ and consider an $n$-user additive-noise mechanism in which the noise vector $Z=(Z_1,\dots,Z_n)\sim P$ and the mechanism satisfies $\varepsilon$-DP for sensitivity $\Delta$. Then
\begin{align}\label{eq: impos}
\int_{\mathbb{R}^n} L\left(\sum_{i=1}^n z_i\right)\,P(dz_1\cdots dz_n)
\;\ge\; C^*_{\varepsilon,\Delta}(L)
\end{align}

Moreover, for every $\eta>0$, there exists a probability measure $\widetilde P$ on $\mathbb{R}^n$ such that the additive-noise mechanism with $Z\sim \widetilde P$ satisfies $\varepsilon$-DP for sensitivity $\Delta$ and
\begin{align}\label{eq: achieve}
\int_{\mathbb{R}^n} L\left(\sum_{i=1}^n z_i\right)\,\widetilde P(dz_1\cdots dz_n)
- C^*_{\varepsilon,\Delta}(L) \;\le\; \eta.
\end{align}

\end{theorem}

This result shows that by using a correlated local differential privacy mechanism, we can arbitrarily approximate  the optimal cost of the central DP setting, while still preventing the server from learning individual data.

Before proceeding with the proof, we introduce some notation.
Let $S$ be an orthogonal matrix whose first row is $\frac{1}{\sqrt{n}}(1,\ldots,1)$, and define $U = SZ$. Under this transformation,\vspace{-2mm}
\[
U_1 = \frac{1}{\sqrt{n}}(Z_1 + \cdots + Z_n),
\]
while the remaining coordinates span the subspace orthogonal to the all-ones vector.

The distribution of $U$ is denoted by $P^{S}$ and is obtained from $P$ via the change of coordinates induced by $S$:
\[
P^{S}(A) := P(S^{-1}A),
\qquad A \subseteq \mathbb{R}^n \text{ measurable},
\]
where $S^{-1}A=\{S^{-1}u:\,u\in A\}$.
With this notation in place, we prove the lower bound in \eqref{eq: impos}.
\begin{proof}[Proof of the lower bound]
Under the change of variables $U=SZ$, the cost depends only on the first coordinate:

\begin{align}\label{eq: cost1D} \int_{\mathbb{R}^n}\! L\left(\sum_{i=1}^n z_i\right)P(dz) &= \int_{\mathbb{R}^n} \!L(\sqrt{n}\,u_1)\,P^{S}(du)\nonumber\\ &= \int_{\mathbb{R}} \!L(\sqrt{n}\,u_1)\,P^{S}_1(du_1), \end{align}
where $P^{S}_1$ is the marginal law of $U_1$ under $P^{S}$. This yields a one-dimensional representation of the cost function.

Next, we express the $\varepsilon$-DP constraint \eqref{eq: epsDP} in the rotated coordinates. The measure $P$ satisfies \eqref{eq: epsDP} if and only if, for all measurable $A\subseteq\mathbb{R}^n$ and all $t\in T_n^\Delta$,
\begin{align}\label{eq: rotated DP}
P^{S}(A) \le e^{\varepsilon} P^{S}(A + St).
\end{align}
Since the first row of $S$ is $\frac{1}{\sqrt n}(1,\dots,1)$, for any $t\in T_n^\Delta$ we have
\[
|(St)_1|
=\frac{\big|\sum_{i=1}^n t_i\big|}{\sqrt n}
\le \frac{\Delta}{\sqrt n}.
\]
Moreover, every $t'\in[-\Delta/\sqrt n,\Delta/\sqrt n]$ can be realized as $(St)_1$ for some $t\in T_n^\Delta$. 
Now fix any measurable $A_1\subseteq\mathbb{R}$ and set $A:=A_1\times\mathbb{R}^{n-1}$. Applying \eqref{eq: rotated DP} to this set and taking the marginal of the first coordinate yields
\begin{align}\label{eq: u1DP}
P^{S}_1(A_1)\le e^{\varepsilon} P^{S}_1(A_1+t'),
\qquad |t'|\le \frac{\Delta}{\sqrt n},
\end{align}
so $P^{S}_1$ satisfies $\varepsilon$-DP for sensitivity $\Delta/\sqrt n$. 
Therefore,
\[
\int_{\mathbb{R}} L(\sqrt n\,u)\,P^{S}_1(du)
\ge C^*_{\varepsilon,\Delta/\sqrt n}\!\big(L(\sqrt n(\cdot))\big)
= C^*_{\varepsilon,\Delta}(L),
\]
where the last equality follows from Lemma~\ref{lem: scaling_cost}. Together with \eqref{eq: cost1D}, this proves the desired lower bound.\qedhere

\end{proof}

\begin{proof}[Proof (sketch) of achievability]

We work in the rotated coordinates $U=(U_1,\dots,U_n)$ to construct $\widetilde P$. Specifically, we first define the coordinate-transformed law $\widetilde P^{S}$ to be a product measure, so that $U\sim \widetilde P^{S}$ has independent components. We choose $U_1$ to have a density close to $f^*_{L(\sqrt n(\cdot)),\varepsilon,\Delta/\sqrt n}$ but satisfying a stricter $\varepsilon_0$-DP constraint for some $\varepsilon_0<\varepsilon$, and take $U_2,\dots,U_n$ to have sufficiently slowly decaying densities. Intuitively, $U_1$ controls the aggregate $\sum_i Z_i$ (and hence the cost), while the remaining coordinates provide enough ``slack'' to absorb the remaining privacy budget $\varepsilon-\varepsilon_0$ under shifts in $ST_n^\Delta$. Consequently, $\widetilde P$ satisfies $\varepsilon$-DP, and letting $\varepsilon_0\uparrow \varepsilon$ yields cost arbitrarily close to the lower bound. A rigorous proof is given in Appendix~\ref{app: ach}.
\end{proof}
\begin{figure}[h]
\centering
\includegraphics[width=0.75\columnwidth]{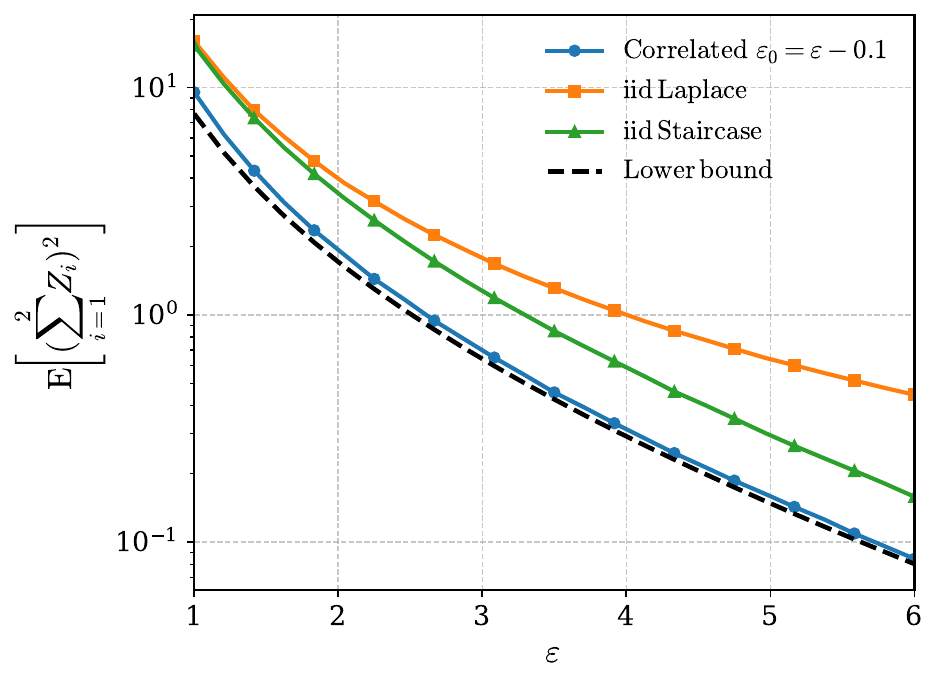}
\caption{Quadratic loss  vs $\varepsilon$ for different DP schemes.}
\label{fig:benchmark}
\end{figure}

We numerically evaluate the two-user case ($n=2$) with sensitivity $\Delta=2$. 
Figure~\ref{fig:benchmark} compares the aggregate squared loss under several mechanisms; the dashed black curve shows the lower bound $C^*_{\varepsilon,\Delta}$.
All other curves are estimated from $5\times 10^5$ samples of $(Z_1,Z_2)$.

The orange curve corresponds to independent $\mathrm{Laplace}(\Delta/\varepsilon)$ noise\footnote{$\mathrm{Laplace}(\lambda)$ has density $\frac{1}{2\lambda}e^{-|x|/\lambda}$, and $\mathrm{Laplace}(\Delta/\varepsilon)$ is $\varepsilon$-DP for sensitivity $\Delta$ \cite{dwork2006calibrating}.}, 
the green curve to the product of one-dimensional optimal densities $f^*_{L,\varepsilon,\Delta}$, and the blue curve to our correlated construction $\widetilde P$ with $\varepsilon_0=\varepsilon-0.1$ (specified in Appendix~\ref{app: ach}). 
These results empirically validate the optimality of the correlated construction relative to standard local DP schemes.

\section{Conclusion}

We showed that the utility gap between local and centralized differential privacy is not inherent to locality, but to the restriction to independent noise. By introducing carefully designed correlations among locally added noise variables, we constructed $\varepsilon$-DP mechanisms whose aggregate cost can approach the optimal centralized cost arbitrarily closely.

Our approach extends beyond pure $\varepsilon$-DP: the lower-bound and scaling arguments carry over to approximate DP and R\'enyi DP \cite{mironov2017renyi}, and achievability should persist when the relevant one-dimensional optimizers are stable under small parameter changes. 

Our analysis also extends to vector-valued data. Under $\ell_\infty$ sensitivity, the problem decouples coordinatewise and reduces to the scalar setting studied here. When sensitivity is measured using other norms, similar lower-bound continue to apply. Moreover, recent results on optimal mechanisms for vector-valued sensitivities \cite{kulesza2025general} suggest that our construction can be extended as well.

An important direction for future work is the efficient and secure generation of the required correlations. In the approximate DP setting, correlated Gaussian schemes use independent noise combined with pairwise anticorrelated components that cancel upon aggregation \cite{bonawitz2017practical,chen2022fundamental,vithana2025correlated}. While similar ideas may be relevant here, our construction is not directly reducible to such a decomposition, as the resulting aggregate noise cannot generally be expressed as a sum of independent random variables.

\bibliographystyle{IEEEtran}
\bibliography{corr}

\newpage
\appendices

\section{Proof of Lemma~\ref{lem: scaling_cost}}\label{app: lemma}

For convenience, we restate Lemma~\ref{lem: scaling_cost}.

\begin{lemma*}[Scaling invariance of the optimal cost]
Let $\varepsilon,\Delta,\alpha>0$ and let $L:\mathbb{R}\to\mathbb{R}$ satisfy \textbf{(A1)}--\textbf{(A3)}. Then
\[
C^*_{\varepsilon,\Delta/\alpha}(L(\alpha(\cdot))) = C^*_{\varepsilon,\Delta}(L).
\]
\end{lemma*}

\begin{proof}

We need to prove the following equality:
\begin{align*}
\int_{\mathbb{R}} L(\alpha x)\, f^*_{L(\alpha(\cdot)),\varepsilon,\Delta/\alpha}(x)dx
=
\int_{\mathbb{R}} L(x) f^*_{L,\varepsilon,\Delta}(x)dx .
\end{align*}

Observe that if a density $g$ satisfies $\varepsilon$-DP for sensitivity $\Delta$, then the scaled density $x\mapsto \alpha g(\alpha x)$ satisfies $\varepsilon$-DP for sensitivity $\Delta/\alpha$. In particular, $\alpha f^*_{L,\varepsilon,\Delta}(\alpha x)$ is feasible for the optimization defining $f^*_{L(\alpha(\cdot)),\varepsilon,\Delta/\alpha}$.

By optimality,\vspace{-1mm}
\begin{align*}
\int_{\mathbb{R}} L(\alpha x) f^*_{L(\alpha(\cdot)),\varepsilon,\Delta/\alpha}(x)dx
& \le \!
\int_{\mathbb{R}}\! L(\alpha x)\alpha f^*_{L,\varepsilon,\Delta}(\alpha x) dx \\
&=
\int_{\mathbb{R}} L(u) f^*_{L,\varepsilon,\Delta}(u)du,
\end{align*}
where the last equality follows from the change of variables $u=\alpha x$.

Applying the same argument with $\alpha$ replaced by $1/\alpha$ yields the reverse inequality, which proves \eqref{eq: scaling_cost}.
\end{proof}

\section{Proof of achievability of Theorem \ref{thm: main}} \label{app: ach}

For the reader's convenience, we restate the achievability claim in Theorem~\ref{thm: main}:

\emph{For every $\eta>0$, there exists a probability measure $\widetilde P$ on $\mathbb{R}^n$ such that the additive-noise mechanism with $Z\sim \widetilde P$ satisfies $\varepsilon$-DP for sensitivity $\Delta$ and}
\begin{align}\label{eq: achieve2}
\int_{\mathbb{R}^n} L\!\left(\sum_{i=1}^n z_i\right)\,\widetilde P(dz_1\cdots dz_n)
- C^*_{\varepsilon,\Delta}(L) \;\le\; \eta.
\end{align}

To prove achievability, we construct a probability measure $\widetilde P$ on $\mathbb{R}^n$ that satisfies $\varepsilon$-DP for sensitivity $\Delta$ and whose cost approaches the lower bound in \eqref{eq: impos}. 

Our construction relies on the \emph{staircase} densities of \cite{geng2015optimal}, named for their piecewise-constant form. They are parametrized by $(\varepsilon,\Delta)$ and a shape parameter $\gamma\in[0,1]$. For notational simplicity, we suppress $(\varepsilon,\Delta)$ and write $f_\gamma$ for $f_\gamma^{(\varepsilon,\Delta)}$:
\begin{align}\label{eq: stair}
f_\gamma(x)\! =\!
\begin{cases}
a_\gamma, & \!\!\!\! 0 \le x < \gamma\Delta,\\
e^{-\varepsilon} a_\gamma, &\!\!\!\! \gamma\Delta \le x < \Delta,\\
e^{-k\varepsilon} f_\gamma(x-k\Delta), 
&\!\!\!\! k\Delta \le x \!< (k+1)\Delta, k\in\mathbb{N},\\
f_\gamma(-x), &\!\!\!\! x<0,
\end{cases}
\end{align}
where $a_\gamma:=a(\gamma,\varepsilon,\Delta)$ is the normalizing constant. Under Assumptions \textbf{(A1)}--\textbf{(A3)}, \cite{geng2015optimal} shows that the optimizer $f^*_{L,\varepsilon,\Delta}$ belongs to this family.

Armed with this characterization, we now prove achievability.
\begin{proof}
Let $\gamma^*$ be such that
\begin{align}\label{eq: opt_gamma}
    f^*_{L(\sqrt n(\cdot)),\varepsilon,\Delta/\sqrt n}
    =
    f_{\gamma^*}^{(\varepsilon,\Delta/\sqrt n)} .
\end{align}
That is, $\gamma^*$ indexes the staircase density that is optimal for privacy level $\varepsilon$ and sensitivity $\Delta/\sqrt{n}$.

Fix $\varepsilon_0\in(0,\varepsilon]$. 
Define $\widetilde P$ via its pushforward under the orthogonal transform $S$ as
\[
\widetilde P^{S}(d(u_1,\dots,u_n))
=
f_{\gamma^*}^{(\varepsilon_0,\Delta/\sqrt n)}(u_1)\, g(u_2)\cdots g(u_n),
\]
where $g$ is a density on $\mathbb{R}$ satisfying $\frac{\varepsilon-\varepsilon_0}{n-1}$-DP for sensitivity $\Delta$. For example, one may take $g$ to be $\mathrm{Laplace}(\lambda)$ with $\lambda=\frac{\Delta(n-1)}{\varepsilon-\varepsilon_0}$.\footnote{In our implementation of $\widetilde P$ in Figure \ref{fig:benchmark}, we take $g$ to be Laplace with this choice of $\lambda$.}

First, we need to verify that this choice satisfies the $\varepsilon$-DP constraint in \eqref{eq: rotated DP}. Observe that for any shift
\begin{align}\label{eq: cuboid}
    \widetilde t\in[-\Delta/\sqrt n,\Delta/\sqrt n]\times[-\Delta,\Delta]^{n-1},
\end{align}
we have
\begin{align*}
& f_{\gamma^*}^{(\varepsilon_0,\Delta/\sqrt n)}(u_1)\, g(u_2)\cdots g(u_n)\\
& \quad \le
e^{\varepsilon_0}\, f_{\gamma^*}^{(\varepsilon_0,\Delta/\sqrt n)}(u_1+\widetilde t_1)\;
\prod_{i=2}^n e^{\frac{\varepsilon-\varepsilon_0}{n-1}}\, g(u_i+\widetilde t_i) \\
& \quad =
e^{\varepsilon}\, f_{\gamma^*}^{(\varepsilon_0,\Delta/\sqrt n)}(u_1+\widetilde t_1)\, g(u_2+\widetilde t_2)\cdots g(u_n+\widetilde t_n).
\end{align*}

Let $ST_n^\Delta$ denote the image of $T_n^\Delta$ under $S$, i.e.,
$ST_n^\Delta := \{Sx:\,x\in T_n^\Delta\}$.
If we show that $S T_n^\Delta$ is contained in the cuboid in \eqref{eq: cuboid}, then the $\varepsilon$-DP claim follows from the characterization in \eqref{eq: rotated DP}.

Indeed, by \eqref{eq: u1DP} we have $(St)_1\in[-\Delta/\sqrt n,\Delta/\sqrt n]$ for all $t\in T_n^\Delta$. Moreover, since $\|t\|_2\le \Delta$ and $S$ is orthogonal, $\|St\|_2=\|t\|_2\le \Delta$, which implies $|(St)_i|\le \Delta$ for every $i\ge 2$. Hence $ST_n^\Delta \subset [-\Delta/\sqrt n,\Delta/\sqrt n]\times[-\Delta,\Delta]^{n-1}$ and therefore $\widetilde P$ satisfies $\varepsilon$-DP for sensitivity $\Delta/\sqrt{n}$.

We now compute the cost under $\widetilde P$. As in \eqref{eq: cost1D},
\[
\int_{\mathbb{R}^n} L\big(\sum_{i=1}^n z_i\big)\,\widetilde P(dz)
=
\int_{\mathbb{R}} L(\sqrt n\,u)\, f_{\gamma^*}^{(\varepsilon_0,\Delta/\sqrt n)}(u)\,du.
\]

It remains to show that, for a suitable choice of $\varepsilon_0$, the right-hand-side integral approximates $C^*_{\varepsilon,\Delta}(L)$. Recall that
\begin{align*}
   C^*_{\varepsilon,\Delta}(L)
   &= C^*_{\varepsilon,\Delta/\sqrt n}(L(\sqrt n(\cdot)))\\
   &= \int_{\mathbb{R}} L(\sqrt{n}\,u)\, f_{\gamma^*}^{(\varepsilon,\Delta/\sqrt n)}(u)\,du,
\end{align*}
where the first equality follows from Lemma~\ref{lem: scaling_cost}, and the second one follows from \eqref{eq: opt_gamma}.

We prove the approximation via the dominated convergence theorem. By \eqref{eq: stair}, 
$f_{\gamma^*}^{(\varepsilon_0,\Delta/\sqrt n)}$ converges to $f_{\gamma^*}^{(\varepsilon,\Delta/\sqrt n)}$ pointwise as $\varepsilon_0\to\varepsilon$, and hence
$L(\sqrt n\ (\cdot) )\,f_{\gamma^*}^{(\varepsilon_0,\Delta/\sqrt n)}$ converges to
$L(\sqrt n (\cdot))\,f_{\gamma^*}^{(\varepsilon,\Delta/\sqrt n)}$ pointwise.

Moreover, the staircase characterization \eqref{eq: stair} implies a uniform envelope: for
$\varepsilon_0\in(\varepsilon/2,\varepsilon]$ the densities $f_{\gamma^*}^{(\varepsilon_0,\Delta/\sqrt n)}$ are dominated by an even function $h$ with geometric tail decay. Since $L$ has at most subexponential growth, we have
$\int_{\mathbb{R}} L(\sqrt n\,u)\,h(u)\,du<\infty$.
Therefore, dominated convergence yields
\begin{align}\label{eq: convergence}
    &\int_{\mathbb{R}} L(\sqrt n\,u) f_{\gamma^*}^{(\varepsilon_0,\Delta/\sqrt n)}(u)du\nonumber\\
& \quad \longrightarrow 
\int_{\mathbb{R}} L(\sqrt n\,u)\, f_{\gamma^*}^{(\varepsilon,\Delta/\sqrt n)}(u)du
\quad \text{as }\varepsilon_0\to\varepsilon,
\end{align}
and hence we can choose $\varepsilon_0$ sufficiently close to $\varepsilon$ so that the absolute difference is at most $\eta$.
\end{proof}
We visualize the asymptotic convergence of $\widetilde{P}$ by numerically simulating two users ($n = 2$) with $\Delta = 2$. The samples of $Z$ are generated where the choice for $g$ is a $\mathrm{Laplace}(\frac{\Delta}{\varepsilon-\varepsilon_0})$.
 We plot the quadratic loss versus $\varepsilon$ for various choices of $\varepsilon_0$ and observe that as $\varepsilon_0$ approaches $\varepsilon$, the quadratic loss decreases and is closer to the theoretical lower bound $C^*_{\varepsilon,\Delta}(L)$.

\begin{figure}[h]
\centering
\includegraphics[width=0.7\columnwidth]{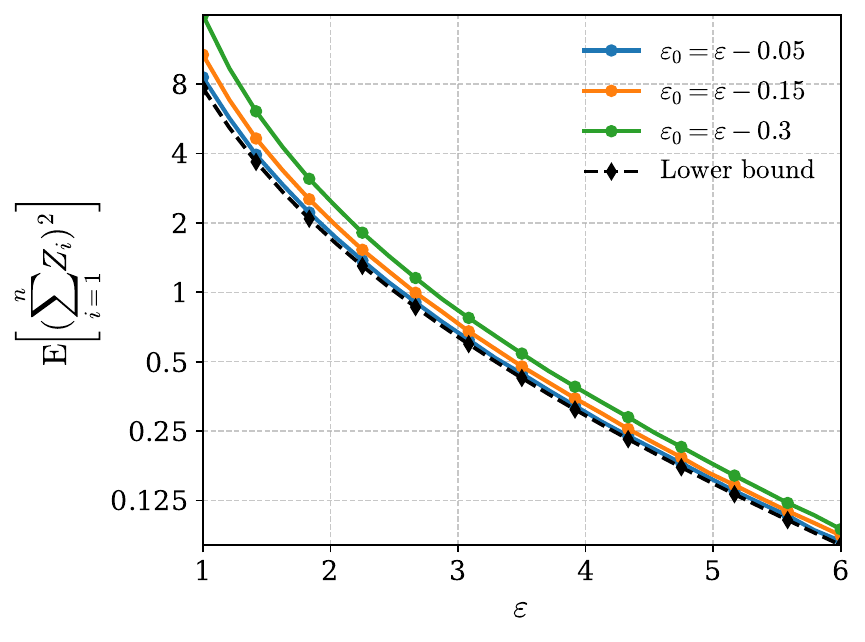}
\caption{quadratic loss vs $\varepsilon$ for different $\varepsilon_0$.}
\label{fig:conv1}
\end{figure}

Figure~\ref{fig: correlated} illustrates the density of $\widetilde P$ for $n=2$.
\begin{figure}[h]
\centering
\includegraphics[width=0.91\columnwidth]{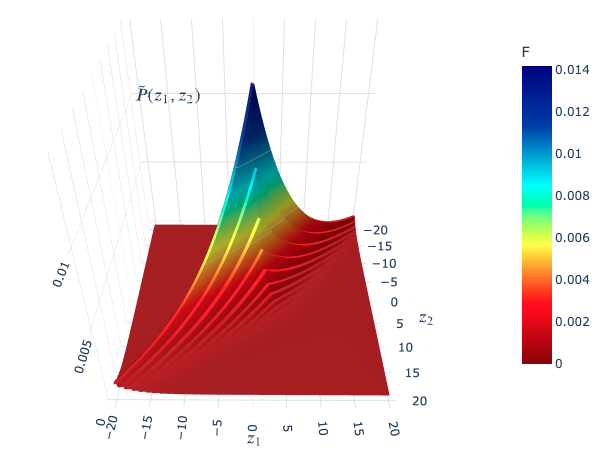}
    \caption{Density of $\Tilde P$}
\label{fig: correlated}
\end{figure}

The plot corresponds to the two-user case with $\varepsilon=0.5$, $\Delta=1$, and $\varepsilon_0=0.4$, where $g$ is chosen as $\mathrm{Laplace}\!\big(\Delta/(\varepsilon-\varepsilon_0)\big)$. As expected, the mass concentrates near the line $z_1+z_2=0$ (i.e., $u_1=0$) and exhibits staircase decay along $z_1=z_2$, equivalently along the $u_1$ direction.

\end{document}